\documentclass[aps,pra,twocolumn,showpacs,amsmath,amssymb,floatfix,superscriptaddress,notitlepage]{revtex4-1}

\usepackage{color}
\usepackage{bbm}
\usepackage{graphicx}
\usepackage{dcolumn}
\usepackage[utf8]{inputenc}

\usepackage[shortlabels]{enumitem}
\usepackage{graphicx}
\usepackage{epstopdf}
\usepackage{amsthm}
\usepackage{amsmath}
\usepackage{empheq}
\usepackage{bbm}
\usepackage{braket}
\usepackage{amssymb}
\usepackage[usenames,dvipsnames]{xcolor}
\usepackage{cases}
\usepackage{latexsym}
\usepackage[colorlinks=true,citecolor=Cerulean,linkcolor=RubineRed,urlcolor=Cerulean]{hyperref}

\graphicspath{ {images/}{images/figure1/}{images/figure2/}{images/figure3/}{images/figure4/}{images/figure5/} }

\newtheorem{theorem}{Theorem}

\newtheorem{lemma}{Lemma}
\newtheorem{corollary}{Corollary}

\newtheorem{definition}{Definition}

\newcommand{\ketbra}[2]{|#1\rangle\langle #2|}

\renewcommand{\v}[1]{\ensuremath{\boldsymbol #1}}

\newcommand{\id}{\mathbbm{1}} 
\usepackage{lipsum}
\usepackage{graphicx}

\newcommand{\beq}{\begin{equation}}
\newcommand{\enq}{\end{equation}}
\newcommand{\beqst}{\begin{equation*}}
\newcommand{\enqst}{\end{equation*}}
\newcommand{\beqar}{\begin{eqnarray}}
\newcommand{\enqar}{\end{eqnarray}}
\newcommand{\beqarst}{\begin{eqnarray*}}
\newcommand{\enqarst}{\end{eqnarray*}}
\newcommand{\beit}{\begin{itemize}}
\newcommand{\enit}{\end{itemize}}

\newcommand*{\cL}{\mathcal{L}}

\newcommand*{\cB}{\mathcal{B}}

\mathchardef\mhyphen="2D

\newcommand{\suppress}[1]{}

\newcommand {\br} [1] {\ensuremath{ \left( #1 \right) }}

\newcommand {\entz} [1] {\mathrm{S}_{0}(#1)}

\newcommand{\sr}[1]{\mathrm{SR}\br{#1}}

\definecolor{henrik}{rgb}{1,.4,0}
\definecolor{purple}{rgb}{1,0,1}

\newcommand{\mc}[1]{\mathcal{#1}}

\newcommand{\e}{\mathrm{e}}
\newcommand{\ii}{\mathrm{i}}

\newcommand{\Tr}{\mathrm{Tr}} 
\newcommand{\supp}{\mathrm{supp}}

\newcommand{\poly}{\mathrm{poly}}

\newcommand{\norm}[1]{\left\Vert #1 \right\Vert}

\newcommand{\braketn}[2]{\langle #1 | #2 \rangle}

\newcommand{\proj}[1]{\ketbra{#1}{#1}}

\newcommand{\LR}{\mathrm{LR}}

\begin{document}

\title{Revivals imply quantum many-body scars}
\date{\today}

\author{Alvaro M. Alhambra}
\affiliation{Perimeter Institute for Theoretical Physics, Waterloo, ON N2L 2Y5, Canada}
\author{Anurag Anshu}
\affiliation{Perimeter Institute for Theoretical Physics, Waterloo, ON N2L 2Y5, Canada}
\affiliation{Institute for Quantum Computing and Department of Combinatorics and Optimization, University of Waterloo, Waterloo, ON N2L 3G1, Canada}

\author{Henrik Wilming}
\affiliation{Institute  for  Theoretical  Physics,  ETH  Zurich,  8093  Zurich,  Switzerland}
\begin{abstract}
	We derive general rigorous results relating revivals in the dynamics of quantum many-body systems to the entanglement properties of energy eigenstates. For a D-dimensional lattice system of N sites initialized in a low-entangled and short-range correlated state, our results show that a perfect revival of the state after a time at most $O(\poly(N))$ implies the existence of at least $O(\sqrt{N}/\log^{2D}(N))$ ``quantum many-body scars": energy eigenstates with energies placed in an equally-spaced ladder and with R\'enyi entanglement entropy at most $O(\log(N)) + O(|\partial A|)$ for any region $A$ of the lattice. This shows that quantum many-body scars are a necessary consequence of revivals, independent of particularities of the Hamiltonian leading to them.  
We also present results for approximate revivals, for revivals of expectation values of observables and prove that the duration of revivals of states has to become vanishingly short with increasing system size.
\end{abstract}
\maketitle
The behaviour of out of equilibrium quantum many-body systems has been gathering a large amount of attention in recent years.
This has largely been motivated by the recent progress of experimental platforms such as cold atoms, ion traps or Rydberg atoms, where many of these systems can be realized in practice \cite{blatt2012quantum,gross2017quantum,bernien2017probing}.
One of the most widely studied situations in this context is that of ``quantum quenches": The system is first prepared in an initial pure state, to then be subjected to an instantaneous change of Hamiltonian $H_0 \rightarrow H$ that drives it out of equilibrium. In generic cases, it is believed that the dynamics will relax to an equilibrium state locally indistinguishable from a thermal ensemble, as granted by the Eigenstate Thermalization Hypothesis (ETH) \cite{Srednicki1994,Srednicki99}. Both the ETH and this relaxing behaviour have been confirmed in numerous numerical and experimental works \cite{d2016quantum,gogolin2016equilibration}. However, there are various cases where this prediction fails notoriously. They include integrable systems that relax to a so-called generalized Gibbs ensemble \cite{vidmar2016generalized}, and also many-body localized systems \cite{nandkishore2015many}, characterized by the presence of quasi-local integrals of motion \cite{imbrie2017local} which prevent the system from thermalizing due to memory of the initial conditions. 

Recently, a new kind of deviation from the predictions of the ETH has been found. It consists of systems which, rather than relaxing, actually revive back to the initial state after a short time. This phenomenon was first found in the experiment of Ref.~\cite{bernien2017probing}, which showed that a system of 51 Rydberg atoms did not thermalize as expected when prepared in a particular initial product state. Shortly after, this was associated with the presence of a number of anomalous energy eigenstates in the spectrum \cite{Turner2018}, the so called \emph{quantum many-body scars}. The first class of models displaying such anomalous eigenstates had been constructed in \cite{shiraishi2017systematic}, and since then, numerous recent efforts have aimed to characterize these eigenstates \cite{ho2019periodic,lin2019exact,khemani2019signatures,iadecola2019quantum,shiraishi2019connection,surace2019lattice,lin2019slow}. Since their discovery, they have been found in further classes of models, see for example Refs.~\cite{vafek2017entanglement,Turner2018a,moudgalya2018entanglement,ok2019topological,bull2019systematic,james2019nonthermal,moudgalya2019quantum,iadecola2019exact,hudomal2019quantum}, (including driven ones \cite{pai2019dynamical,mukherjee2019collapse,haldar2019scars}), some of which even display perfect revivals when the system is prepared in particular product states \cite{choi2019emergent,schecter2019weak} or matrix product states (MPS) \cite{Chattopadhyay2019,Iadecola2019kinetic}. While it is clear that in any model exhibiting scarred eigenstates there are relatively low-entangled initial states that show perfect revivals (simply take a super-position of two scarred eigenstates), it is not expected that one can always find short-range correlated states (e.g., product states) that show perfect revivals.

Motivated by these recent findings, we here derive a number of analytical results that apply to many-body systems exhibiting revivals at short times from low-entangled and short-range correlated states. Our results significantly improve on a Lemma presented in Ref.~\cite{choi2019emergent}.  We first derive properties of the energy spectrum and eigenstates that have to be fulfilled whenever (approximate) revivals appear in a local quantum many-body system, independent of the details of the Hamiltonian and in any dimension $D$ of the underlying lattice. We show that the existence of at least $\mc O(\sqrt{N}/\log^{2D}(N))$ (where $N$ is the system size) quantum many-body scars follows from the early revivals of low-entangled and short-range correlated initial states, when the revival time $\tau$ is at most of the order of $\poly(N)$. 
We prove that all of these quantum many-body scars have R\'enyi entanglement entropies (of orders $\alpha>1$) of at most $\mc O(\log(N))+O(|\partial A|)$, for any subset $A$ of the lattice sites, with the area law term vanishing if the initial state experiencing revivals is a product state. Our bounds hence match the scaling that has been found in concrete model Hamiltonians \cite{Turner2018a,vafek2017entanglement,moudgalya2018entanglement,choi2019emergent,schecter2019weak,Chattopadhyay2019,Iadecola2019kinetic}. In dimension $D=2$ or higher and for initial product states, our results show that quantum-many-body scars show even weaker entanglement in terms of R\'enyi entropies of order $\alpha>1$ than allowed by an area law (with $\log$ corrections). For R\'enyi entropies with $\alpha \le 1$, we use techniques inspired by the problem of bounding entanglement of ground states of gapped models \cite{arad2013area} to show that the entanglement entropy scales at most $\mathcal{O}(\sqrt{N | \partial A |})$.

The paper is structured as follows: In Section \ref{sec:revival} we state our assumptions on the initial states and define the notion of exact and approximate revivals. Then, in Section \ref{sec:constraints}, we give constraints on the energy distribution of initial states with revivals, which we use to give bounds on the R\'enyi entanglement entropy $S_\alpha$, first in Section \ref{sec:entanglement} (for $\alpha>1$) and then in Section \ref{sec:entanglement2} ($\alpha \le 1$). In Section \ref{sec:observable} we explain the consequences of perfect revivals on an observable, and in Section \ref{sec:universal} we give universal constraints that all periodic revivals must obey. In the Appendix we include the proof of some of the technical statements and discuss a model example to benchmark our bounds.

\section{Systems with revivals}\label{sec:revival}
We consider a system on a regular $D$-dimensional lattice $\Lambda$ of $N$ $d$-dimensional sites with a local Hamiltonian $H=\sum_{x\in \Lambda} h_x$, where the local terms $h_x$ have support on at most $b$ neighboring sites, and are uniformly bounded by a constant $h$: $\norm{h_x}\leq h$. As usual, we denote the unitary implementing time evolution by $U_t=\exp(-\ii H t)$ and the energy eigenstates by $\ket{E_j}$. Without loss of generality, we assume that the ground state energy vanishes, $E_0=0$, and set $\hbar=1$. 
We further assume that the system is prepared in a pure state $\ket{\Psi}$ which:
\begin{enumerate}[i)]
	\item \label{ass:entanglement} Is a \emph{low-entangled} state: For every region $A$ on the lattice, the reduced density matrix $\sigma_A$ has rank at most $\chi^{|\partial A|}$, where $|\partial A|$ is the area of the boundary of $A$ and $\chi$ is independent of the system size.
	\item \label{ass:clustering} Is \emph{short-range correlated} and out of equilibrium: It fulfills exponential decay of correlation with a finite correlation length and the standard deviation of the energy is given by $\sigma \equiv \sqrt{\langle H^2 \rangle - \langle H\rangle^2} = s \sqrt{N}$ for some constant $s>0$.
\end{enumerate}
The statement of assumption \ref{ass:entanglement} is somewhat technical, but it includes all states that can be represented by a tensor network with constant bond dimension, such as projected entangled-pair states (PEPS) in D=2 and matrix product states (MPS) in D=1. 
The constant $\chi$ is then directly related to the bond dimension. In particular, for product states we have $\chi=1$. 
The upper bound $\sigma\leq s\sqrt{N}$ required in assumption \ref{ass:clustering} follows directly from the finite correlation length. The assumption therefore simply makes explicit that the initial state must not be an eigenstate of the Hamiltonian. We emphasize that generic tensor network states also have a finite correlation length \cite{Lancien2019}.

The way in which we understand revivals of a state is in terms of the fidelity with the initial state, captured by the following definition.
\begin{definition}
	An initial state $\ket{\Psi}=\sum_j c_j \ket{E_j}$ evolved with a Hamiltonian $H$ has an $\epsilon$-revival at time $\tau$ if 
	\begin{equation}
	|F(\tau) - F(0)|\leq \epsilon,
	\end{equation}
	where $F(t) = |f(t)|$ with $f(t)=\bra{\Psi} e^{-\ii t H} \ket{\Psi}=\sum_i \vert c_j\vert^2 e^{-\ii t E_j}$.
\end{definition}	
The definition only involves an $\epsilon$-revival at a single time $\tau$. However, it implies that there are further periodic approximate revivals at later times. Concretely, an $\epsilon$-revival at time $\tau$ implies (see appendix~\ref{sec:revivals} for derivation)
\begin{align}\label{eq:revivals}
    F(m\tau) \geq 1 - m\sqrt{2\epsilon},\quad m\in\mathbb N.
\end{align}
We emphasize that the revival of the full many-body state is a very strong condition and $f(t)$, sometimes known as ``spectral form factor" and its absolute $F(t)$ value as ``survival probability", is not a directly measurable quantity ($F(t)$ is, however, measurable in principle using an interferometric Ramsey scheme \cite{de2008ramsey,goold2011orthogonality,knap2012time,cetina2016ultrafast}).
For this reason,  we also consider the case of a perfectly recurring expectation value of an observable $A$, leading to similar results under an additional assumption (see Section~\ref{sec:observable}).

\section{Constraints on the energy distribution}\label{sec:constraints}
\begin{figure}[t]
	\centering
	\includegraphics[width = 0.99\linewidth]{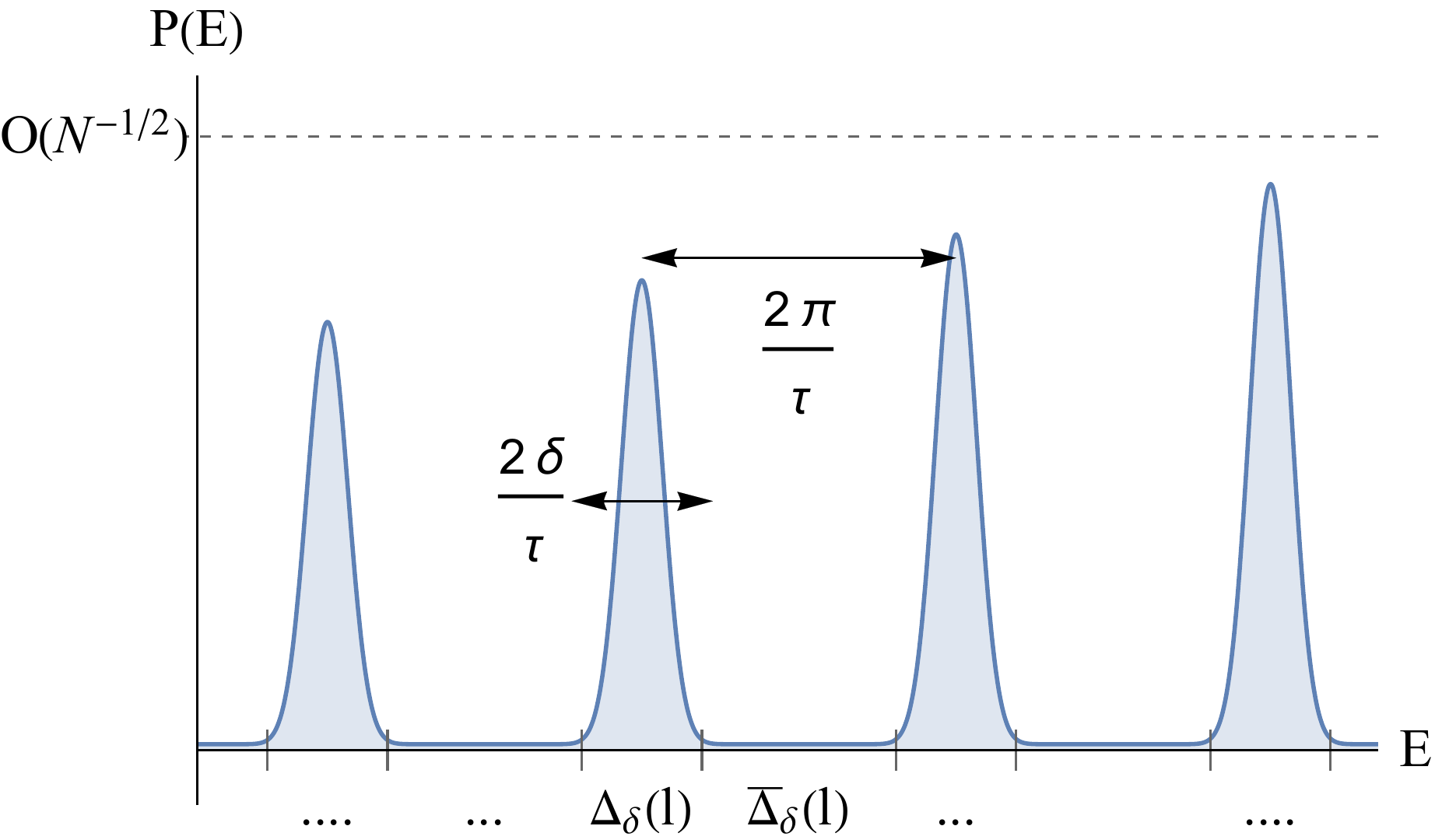}
	\caption{
	Schematic diagram of the energy distribution and of the intervals $\Delta_\delta(l)$, with equally-spaced peaks of width $2\delta /\tau$. In the limit of perfect revivals, the peaks have width $0$. The upper bound follows from the Berry-Esseen theorem (Theorem \ref{th:berry} in the Appendix). Theorem \ref{th:nc} guarantees that at least $\mathcal{O}(\sqrt{N}/\log^D(N))$ peaks have a weight larger than $\mathcal{O}(1/\text{poly}(N) )$.} 
	\label{fig:peaks}
\end{figure}
From the definition $f(t) = \bra{\Psi}\exp(-\ii H t)\ket{\Psi}$, it is clear that $f(t)$ is the characteristic function of the probability distribution of energy. 
In this section, we therefore study the properties of the probability distribution of energy in the case of $\epsilon$-revivals of a state that fulfils assumption \ref{ass:clustering}  above.
First we show that if there are approximate revivals at short times $\tau$, a large weight of the distribution is contained within equally-spaced ``peaks", whose spacing depends on $\tau$ (see Fig.~\ref{fig:peaks}). This is true for any initial state.
We then make use of the fact that the probability distribution of energy of a state with a finite correlation length is roughly Gaussian, with which we show that at least $\sim \sqrt{N}$ of the peaks each contain total weight of at least $\mc O(1/N)$. 

As before, we write the initial state in the energy eigenbasis as $$\ket\Psi = \sum_j c_j \ket{E_j}.$$ In the case where the Hamiltonian has degenerate energy-levels, we choose the basis in each energy eigenspace so that every energy appears only once in the above decomposition of $\ket\Psi$.
To set up some further notation, let us introduce $\alpha(t)$ as the phase of $f(t)$:
\begin{align}
	f(t) = \e^{\ii \alpha(t)}F(t),\quad \alpha(0)=0. 
\end{align}
Given $\alpha(\tau)$ and an arbitrary constant $0\leq \delta\leq\pi$, we define for all $l\in\mathbb Z$ the energy intervals 
\begin{align}
	\Delta_\delta(l) = \frac{2\pi l+ \alpha(\tau)}{\tau} + \left[-\frac{\delta}{\tau}, \frac{\delta}{\tau}\right] ,
\end{align}
where addition is point-wise, and the interval between two consecutive such intervals
\begin{align}
	\overline{\Delta}_\delta(l) = \left(\frac{2\pi l+\alpha(\tau)}{\tau} + \frac{\delta}{\tau},\frac{2\pi (l+1)+\alpha(\tau)}{\tau} - \frac{\delta}{\tau} \right).
\end{align}
These partition the real line as
$\mathbb R = \bigcup_{l\in\mathbb Z} \Delta_\delta(l)\cup \overline{\Delta}_\delta(l)$.
Note that since $\vert \vert H \vert \vert \leq  h N$, the number of intervals $\Delta_\delta(l)$ in the spectrum with nonzero energy eigenvalues is at most
\begin{align}
n \equiv \frac{\tau h}{2\pi} N\propto \tau N.
\end{align}
To de-clutter the notation in what follows, let us also introduce $p$ as the probability measure of energy of the initial state, so that 
\begin{align}
	p(\Delta_\delta(l)) = \sum_{i:E_i \in \Delta_\delta(l)} |c_i|^2. 
\end{align}

The following lemma lower bounds the probability of measuring an energy on the initial state within one of the intervals $\Delta_\delta(l)$.
\begin{lemma}\label{le:recurrence}
	Let $\epsilon \ge |F(\tau) - F(0)|$. Then
\begin{align}\label{eq:bound1}
\sum_{l\in \mathbb Z} p(\Delta_\delta(l))\geq 	1 -\frac{\epsilon}{1-\cos(\delta)}.
\end{align}
\end{lemma}
\begin{proof}
The proof follows from simple applications of inequalities between complex numbers.	Since $f(0)=1$, we have
	\begin{align}
		|F(\tau)-F(0)| &= \left|\e^{-\ii \alpha(\tau)} f(\tau) -f(0)\right| \\
		&=\left| \sum_j |c_j|^2 \left(1- \e^{-\ii \alpha(\tau) - \ii E_j\tau}\right)\right|.
	\end{align}
	Since for any complex number we have $|z| = \sqrt{\mathrm{Re} z^2 + \mathrm{Im} z^2} \geq |\mathrm{Re}z|$, we get the lower bound
\begin{align}
	|F(\tau)-F(0)|\geq \sum_j |c_j^2|(1-\cos(E_j \tau+\alpha(\tau))).
\end{align}
	We now split up the summation in terms of the intervals $\Delta_\delta(l)$ and $\overline{\Delta}_\delta(l)$ and neglect the contributions from $\Delta_\delta(l)$. This yields a lower bound
\begin{align}
	|F(\tau)-F(0)| &\geq 
	\sum_{l\in \mathbb Z} \sum_{j: E_j \in \overline{\Delta}_\delta(l)}|c_j|^2(1-\cos(E_j\tau+\alpha(\tau))).
\end{align}
	For $E_j \in \overline{\Delta}_\delta(l)$  we have that
\begin{align}
	\cos(E_j \tau+\alpha(\tau)) \leq \cos(\delta). 
\end{align}
We hence obtain
\begin{align}
	&\epsilon \geq |F(\tau) - F(0)| \geq (1-\cos(\delta))\sum_{l\in \mathbb Z}p\left(\overline{\Delta}_\delta(l)\right)\nonumber
	\\& \Rightarrow\quad \sum_{l\in \mathbb Z}p\left(\overline{\Delta}_\delta(l)\right)\leq \frac{\epsilon}{1-\cos(\delta)}.
\end{align}
Using the normalization of the probability distribution of energy, we then find 
\begin{align}
	1 -\frac{\epsilon}{1-\cos(\delta)}\leq  \sum_{l\in \mathbb Z} p\left(\Delta_\delta(l)\right). 
\end{align}
\end{proof}

Lemma \ref{le:recurrence} tells us that if an $\epsilon$-revival at time $\tau$ occurs, the energy distribution must be mostly contained in the intervals $\Delta_\delta(l)$ as long as $\cos(\delta)$ is not too close to unity. 
The smaller $\epsilon$ (which is equivalent to an increasingly exact revival), the narrower the intervals $\Delta_\delta(l)$ can be made, by choosing a $\delta$ such that the RHS of \eqref{eq:bound1} is close to $1$.
If the recurrence time $\tau$ is very large, both the distance between the intervals $\Delta_\delta(l)$ and their width $2\delta/\tau$ is small. 
In a finite system, for every $\epsilon>0$, recurrence theorems guarantee \cite{bocchieri1957quantum,wallace2013recurrence} the existence of a corresponding recurrence time $\tau_{R}$. For generic systems, however, one expects that$\tau_R = \mc O(\exp(\exp(N)))$, while for particular cases such as integrable systems, it is expected that $\tau_R = \mc O(\exp(N))$ \cite{venuti2015recurrence}.
In either case, the distance between the intervals $\Delta_\delta(l)$ becomes comparable to or smaller than the level-spacing, so that the union of the $\Delta_\delta(l)$ automatically contains (almost) all energy eigenvalues.  

The next important feature of energy distributions of local models in a state with finite-correlation length is given by the Berry-Esseen theorem \cite{brandao2015equivalence}. This is a strengthening of the central limit theorem, in which the error from having finite sample sizes is bounded by a function of the number of samples. It allows us to derive the second key constraint.
\begin{lemma}\label{le:upperbound} Let $\ket{\Psi}$ be a state fulfilling assumption~\ref{ass:clustering}. Then there exists a constant $K\geq 0$ (independent of $N$) such that
\begin{align}
	p(\Delta_\delta(l)) \leq \frac{\delta}{\sigma\tau} + K \frac{\log^{2D}(N)}{\sqrt{N}}.
\end{align}
\end{lemma}
The proof can be found in Appendix \ref{app:upperbound}.
Lemma~\ref{le:recurrence} and Lemma~\ref{le:upperbound} have competing effects: While Lemma \ref{le:recurrence} shows that the distribution clusters around at most $n$ evenly spaced energy intervals, Lemma~\ref{le:upperbound} guarantees that no particular interval of energy width $\delta$ can contain a weight larger than $\delta/(\tau\sigma) + K \log^{2D}(N)/\sqrt{N}$. Together, they imply the existence of a large number of intervals containing each a certain minimum weight:
\begin{theorem}\label{th:nc}
	Given an initial state fulfilling assumption~\ref{ass:clustering} and with an $\epsilon$-revival at time $\tau$, then for every $c>1$ and $0< \delta\leq \pi$, the number $N_{c,\delta}$ of intervals $\Delta_\delta(l)$ in the energy distribution with $p(\Delta_\delta(l))> 1/(cN)$ is lower bounded as
	\begin{align}
	N_{c,\delta} \geq  \sqrt{N}\frac{\left[1 - \frac{h \tau}{2\pi c} - \frac{\epsilon}{1-\cos(\delta)}\right]}{\delta/(\tau s) +  K \log^{2D}(N)}.
	\end{align}
\end{theorem}
\begin{proof}
	The total number of peaks $\Delta_\delta(l)$ is upper bounded by $n= h N \tau/2\pi $. Hence the total number of peaks such that $p(\Delta_\delta(l))\leq 1/(cN)$ is trivially also upper bounded by $n$. Let the index set $J_\delta$ collect the peaks such that $p(\Delta_\delta(l))> 1/(cN)$. Then using Lemma~\ref{le:recurrence} we find
\begin{align}
	1 - \frac{\epsilon}{1-\cos(\delta)} &\leq \sum_{l\notin J_\delta}p(\Delta_\delta(l))+ \sum_{l\in J_\delta}p(\Delta_\delta(l)) \\
&\leq n\frac{1}{cN} + \sum_{l\in J_\delta}p(\Delta_\delta(l)) \\
	&= \frac{h\tau}{2\pi c} + \sum_{l\in J_\delta}p(\Delta_\delta(l)).
\end{align}
	Using Lemma~\ref{le:upperbound} we then get
\begin{align}
	1 - \frac{\epsilon}{1-\cos(\delta)} \leq \frac{h\tau}{2\pi c}+\frac{N_{c,\delta}}{\sqrt{N}}\left(\frac{\delta}{s\tau}  + K \log^{2D}(N)\right)\nonumber
\end{align}
and re-arranging yields the desired bound.
\end{proof}
To understand this bound, let us make a specific choice for $c$ and $\delta$, assuming that $\epsilon$ is very small. For example, we can choose $c = 2 h \tau/\pi$ and $\delta=\sqrt{2\epsilon}$, so that $\epsilon/(1-\cos(\delta))\approx 1/2$ and we find
\begin{align}
	N_{c,\delta} \geq \frac{\sqrt{N}}{4} \left[\frac{\sqrt{2\epsilon}}{\tau s} + K \log^{2D}(N)\right]^{-1}. 
\end{align}
We see that the number of peaks of width $\sqrt{2\epsilon}$ such that each of them contains weight at least $(\pi/2)/(h \tau N)$ is essentially lower bounded by $\mc O(\sqrt{N}/\log^{2D}(N))$.

The result holds for any value of $\tau$, but if $\tau$ scales very quickly with $N$, this result loses its predictive power.  For $\tau=\poly(N)$, one still finds a total weight of $1/\poly(N)$ in each peak, which is sufficient for our arguments on entanglement in the next section.
However, if we consider the usual recurrence time $\tau_R$ for some $\epsilon>0$ in a generic system, which is $\tau_R = \mc O(\exp(\exp(N)))$, our bound trivializes: the r.h.s. becomes negative if we do not choose $c$ doubly-exponentially large in the system size. At the same time, if we choose $c$ doubly-exponentially large, the peaks are only required to contain an doubly-exponentially small amount of weight, which does not yield useful information. 


\section{Bounds on R\'enyi entanglement entropy with $\alpha>1$}\label{sec:entanglement}
We now estimate the entanglement entropy of (approximate) eigenstates of the system. The previous discussion motivates the definition of the following normalized pure states
\begin{align}
\ket{\hat{E}_l} = \frac{1}{\sqrt{p(\Delta_\delta(l))}}\sum_{E_j\in \Delta_\delta(l)} c_j \ket{E_j}.
\end{align}
These are approximate energy eigenstates with energy $\hat E_l = (2\pi l+\alpha(\tau))/\tau$ for which $\delta/\tau$ controls the precision, in the sense that
\begin{align}\label{eq:quasieigenstates}
\norm{H\ket{\hat E_l} - \hat{E_l}\ket{\hat E_l}} 
	\leq \frac \delta \tau
\end{align}
and
\begin{align}\label{eq:quasieigenstate2}
\norm{U_t \ket{\hat E_l} - \e^{-\ii \hat E_l t}\ket{\hat E_l}}\leq \sqrt{2(1-\cos(\delta t/\tau))} \approx \delta \frac{t}{\tau},
\end{align}
where the last approximation holds for $\delta t/\tau \ll 1$. The states $\ket{\hat E_l}$ hence dephase in a time of order $\tau/\delta$, but cannot be distinguished from eigenstates on time-scales much smaller than that. In the limit $\delta/\tau \rightarrow 0$ they converge to actual eigenstates provided that the limit exists, i.e., the interval $\Delta_\delta(l)$ actually contains an eigenstate in this limit.

Theorem \ref{th:nc} implies that the initial state has a fidelity of at least $1/{c N}$ with $N_{c,\delta}$ of the approximate eigenstates. This is in fact enough to bound the R\'enyi entanglement entropy $S_\alpha$ of those approximate eigenstates for every region of the lattice, which is the focus of our next main result. We remind the reader at this point that the R\'enyi entropies are defined as
\begin{align}
	S_\alpha(\rho) = \frac{1}{1-\alpha}\log\left(\Tr[\rho^\alpha]\right).
\end{align}
In the limit $\alpha\rightarrow 1$ they converge to the von~Neumann entropy pointwise and they fulfill $S_\alpha \leq S_\beta$ for $\alpha\geq \beta$. In the limit $\alpha\rightarrow \infty$, one obtains $S_\infty(\rho)=-\log(\norm{\rho})$.

\begin{theorem}
	There exist at least $N_{c,\delta}$ many of the approximate eigenstates $\ket{\hat E_l}$ with the following property: For all $\alpha>1$ and for any subregion $A$ of the lattice, 
	the R\'enyi entanglement entropy is bounded as
	\begin{equation}
		S_{\alpha}(\hat \rho^{(l)}_A) \le \frac{\alpha}{\alpha-1}\left[\log {c N }+|\partial A|\log(\chi)\right],
\end{equation}
	where $N_{c,\delta}$ is bounded as per Theorem \ref{th:nc} and $\hat \rho^{(l)}_A = \Tr_{A^c}[\proj{\hat E_l}]$.
\end{theorem}
\begin{proof}
	This result is a slight extension of an argument from Ref.~\cite{wilming2018entanglement}: 
	Since the fidelity $F$ between two quantum states cannot decrease under tracing out sub-systems, we have
	\begin{align}
		|\braketn{\Psi}{\Phi}|^2 &\leq F(\rho_A,\sigma_A)^2,
	\end{align}
	where, $\rho_A = \Tr_{A^c}[\proj{\Phi}]$ is the reduced state of $\Phi$ on $A$ and $\sigma_A$ that of $\ket{\Psi}$.
	The fidelity between two states is smaller than that of the outcome-distributions of any measurement on the states. We can therefore use the binary projective measurement $\{P_\sigma, \mathbf 1 -P_\sigma\}$, with $P_\sigma$ being the projector onto the image of $\sigma_A$, to find
	\begin{align}
	F(\rho_A,\sigma_A)^2&\leq \Tr[\rho_A P_\sigma]\leq \mathrm{rank}(\sigma_A)\norm{\rho_A}\\
&=\mathrm{rank}(\sigma_A)\exp(-S_\infty(\rho_A)).
	\end{align}
	 In Refs.~\cite{Beck1990,wilming2018entanglement} it was further shown that $S_\infty \geq \frac{\alpha-1}{\alpha} S_\alpha$. Using assumption \ref{ass:entanglement}, we thus find in our case
	\begin{align}
		\frac{1}{cN} \leq |\braketn{\Psi}{\hat E_l}|^2 \leq \chi^{|\partial A|}\exp\left(-\frac{\alpha-1}{\alpha}S_\alpha(\hat \rho_A^{(l)})\right)
	\end{align}
	and solving for $S_\alpha(\hat\rho^{(l)}_A)$ yields the desired bound.
\end{proof}
In $D=1$, $|\partial A|$ simply counts the number of connected components of $A$ and for a product state we have $\chi=1$, so that the area law term vanishes.
As long as $ c= \mc O(\mathrm{poly}(N))$, the result then leads to $O (\sqrt{N}/\poly(\log(N)))$ approximate eigenstates with entanglement entropy of order $\mc O(\log(N))$. Since $\tau < c$, this allows for a longer revival time, of up to $\tau = \mc O(\mathrm{poly}(N))$.  

For systems that exhibit perfect revivals, $\epsilon=0$, we can choose $\delta=0$, so that by Eq. \eqref{eq:quasieigenstates} the $\ket{\hat E_l}$ become exact eigenstates with energies in the set $\{(2\pi k + \alpha(\tau))/\tau\}_{k\in \mathbb Z}$. 
\begin{corollary}\label{co:entanglement}
	If $F(\tau)=F(0)$ for some $\tau$, there exists a set of at least $N_{c,0}$ energy eigenstates $\ket{E_l}$ with energies in the set $\{(2\pi k+\alpha(\tau))/\tau\}_{k\in \mathbb Z}$ and such that their entanglement entropy of any region is bounded as 
\begin{equation}
	S_{\alpha}(\rho^{(l)}_A) \le \frac{\alpha}{\alpha-1}\left[\log {c N}+|\partial A|\log(\chi)\right],\quad\alpha>1,
\end{equation}
where
\begin{align}
N_{c,0} \geq  \sqrt{N}\frac{\left(1 - \frac{h \tau}{2\pi c}\right)}{ K\log^{2D}(N)}.
\end{align}
\end{corollary}

This bound on the entropy is consistent with the examples in \cite{Turner2018a,vafek2017entanglement,moudgalya2018entanglement,choi2019emergent,schecter2019weak,Chattopadhyay2019,Iadecola2019kinetic}, which display eigenstates with a $\log(N)$ scaling of the von Neumann entropy (see Appendix \ref{app:example} for a more detailed comparison). 

\section{Bounds on the R\'enyi entanglement entropy with $\alpha \le 1$} \label{sec:entanglement2}

For this range of entropies, let us again consider the case of perfect revivals, with $F(\tau)=F(0)$. In this case, the scar states on which the initial state $\ket{\Psi}$ has support can be exactly represented by polynomial functions of the Hamiltonian, with a low degree. Consider the polynomial
\begin{equation}
\label{exASSP}
K_i(H) = \prod_{j\neq i}\left (\br{\id - \frac{\br{H-E_i}^2}{\br{E_j-E_i}^2}} \right ),
\end{equation}
where the product over $j$ ranges from $1$ to $h N \tau/2\pi$, except for $i$. Note that $K_i(H)\ket{E_i}=\ket{E_i}$ and $K_i(H)\ket{E_j}=0$. Thus, 
\begin{equation}
\label{eq:exproj}
K_i(H)\ket{\Psi}= c_i\ket{E_i}
\end{equation}
which can be interpreted as the statement that the polynomial $K_i(H)$ projects $\ket{\Psi}$ to the state $\ketbra{E_i}{E_i}$. The next result is an immediate consequence of this construction, following \cite{arad2013area}.

\begin{theorem}
\label{th:S0scar}
If $F(\tau)=F(0)$, then for all eigenstates $\ket{E_i}$ with non-zero support on $\ket{\Psi}$, the R\'enyi-$0$ entanglement entropy of sufficiently regular regions is bounded as
$$S_{0}(\hat \rho^{(l)}_A)\leq 7\sqrt{hN \tau |\partial A|} \log{\left(hN^2 \tau d^b\right)} + \vert \partial A \vert \log{\chi}.$$
\end{theorem}

The proof, together with a precise definition of what we mean by sufficiently regular regions, can be found in Appendix \ref{app:upperbound2}. This result also bounds the von~Neumann entropy, since $S_1 \le S_0$. Notice that this result holds for all eigenstates on the equally-spaced ladder, as opposed to only $\mathcal{O}(\sqrt{N})$ of them as in Corollary \ref{co:entanglement}. While it represents a non-trivial bound, much smaller than the  $\mathcal{O}(N)$ expected for most eigenstates, the concrete models in the literature show that this could potentially be improved to $\mathcal{O}(\log N)$.
Indeed, in concrete models, the scar states can usually be written as $\ket{E_i}= (\sum_j S_j)^i\ket{\Phi}$ for some simple state $\ket\Phi$ (such as the ground state), where $S_j$ are single particle operators and $j$ labels the sites of the lattice. These eigenstates $\ket{E_i}$ have equally-spaced energies $E_i=\omega i +E_0$, so that there are at most $hN/\omega$ of them in the spectrum. Writing
\begin{align}
    (\sum_j S_j)^i
    &= \sum_{k=0}^i{i\choose k}(\sum_{j\in A}S_j)^k(\sum_{j\in A^c}S_j)^{i-k},
\end{align}
we find that the Schmidt rank of the operator $(\sum_j S_j)^i$ is at most $\log(i)\leq \log\frac{hN}{\omega}$. Thus, if $\ket{\Phi}$ is low-entangled in the sense of assumption~\ref{ass:entanglement}, as is usually the case, all entanglement entropies of $\ket{E_i}$ are bounded by $\log\frac{hN}{\omega}+\vert\partial A\vert \log{\chi}$. 

\section{Revivals in an observable}\label{sec:observable}
Assuming a revival of the full many-body state is a rather strong condition. Intuitively, it should be possible that physically relevant observables have a revival in terms of their expectation value  $\langle A(t)\rangle\equiv \bra{\Psi(t)}A \ket{\Psi(t)}$ at time $\tau$ and yet the full many-body state has a small overlap with the initial state, $F(\tau)\ll 1$. 
It may therefore be surprising that conclusions similar to those above can be reached when one assumes that the expectation value is periodic,
\begin{equation}
\langle A(t) \rangle=\langle A(t+ \tau) \rangle \,\, \forall t,
\end{equation}
and makes one further assumption on the observable. To state this assumption, let us write
\begin{align}
	\langle A(t) \rangle = \sum_{i,j} c_i c_j^* A_{ij} e^{-i (E_i-E_j)t}=\sum_{\omega} v_\omega e^{-i \omega t}, 
\end{align}
such that $v_\omega=\sum_{E_i-E_j=\omega} c_i c_j^* A_{ij}$. Then for any $\omega'$,
\begin{align}
0&=\lim_{T \rightarrow \infty} \int_0^T \frac{\text{d}t}{T} \left(\langle A(t+\tau ) \rangle-\langle A(t) \rangle\right) e^{i \omega ' t}\nonumber \\&=(e^{-i \tau \omega'}-1) v_{\omega'}.\nonumber 
\end{align}
It follows that either the frequency $\omega'$ does not appear in the dynamics of the expectation value ($v_{\omega'}=0$) or it is of the form $\omega'=2 \pi l /\tau$. For local observables in many-body systems, we expect that in general $A_{ij}\neq 0$ unless $A$ and $H$ share some symmetry. It therefore seems reasonable to assume that generically 
\begin{align}\label{eq:assumption_A}
	v_\omega=0 \implies c_i c^*_j=0, \ \forall \, E_i-E_j=\omega. 
\end{align}
We thus conclude that $c_i c^*_j\neq 0$ only if $E_i-E_j=2\pi l/\tau$ for some integer $l$. This in turn implies $F(0)=F(\tau)$, which is the assumption of Corollary~\ref{co:entanglement}.
We leave it as an open problem to explore the setting of approximate $\epsilon$-revivals of local observables.

\section{Universal constraints on revivals}\label{sec:universal}
In the preceding sections, we have assumed revivals and derived properties of the energy eigenstates from this assumption. 
Before concluding, let us now briefly discuss general constraints for such revivals which apply to any model with a local Hamiltonian. 
It is expected that if revivals of the initial product state exist, their duration (i.e., the time for which $F(t)$ is larger than some constant) must become vanishingly short in the thermodynamic limit (see Fig. \ref{fig:revivals}). 
This is also a feature found in the concrete models (see, for example, Refs.~\cite{schecter2019weak,Chattopadhyay2019}). 
This property does not imply that the duration of revivals for local observables becomes short in the thermodynamic limit, but only the time-interval in which the full many-body state has large overlap with the initial state.

We illustrate behaviour this with two different and general results. The first one shows that the average fidelity over time decays with the system size for any initial state fulfilling assumption \ref{ass:clustering}.
\begin{theorem}\label{th:fidbound}
	Let $\ket{\Psi}$ be a pure state fulfilling assumption \ref{ass:clustering}. Then
	\begin{equation}\label{eq:avgfidelity}
	\int_0^T \frac{\text{d}t}{T} \left\vert F(t) \right \vert^2 \le \frac{5 \pi }{2\sigma T}+K'  \frac{\log^{2D}(N)}{\sqrt{N}}
	\end{equation}	
	where $K'\geq 0$ is a constant independent of system size.
\end{theorem}
A revival at a time $\tau$ of fidelity at least $(1-\epsilon)$ for a time interval of length $\tau_{\text{rev}}$ contributes to the LHS of Eq. \eqref{eq:avgfidelity} with $(1-\epsilon)\tau_{\text{rev}}/\tau$, so that
\begin{equation}
\frac{(1-\epsilon)\tau_{\text{rev}}}{\tau} \le \frac{5 \pi }{2\sigma \tau}+K'  \frac{\log^{2D}(N)}{\sqrt{N}}.
\end{equation}
For times $\tau \ge \mathcal{O}(1)$, we see that the RHS goes as $\sim O(1/\sqrt{N})$. This bound then restricts the revivals of high fidelity to either a very short time interval $\tau_{\text{rev}}$ or a very late time $\tau$. In Appendix \ref{app:example} we show that the model from \cite{schecter2019weak} effectively saturates this bound. 

\begin{figure}[t]
    \includegraphics[width=0.8\linewidth]{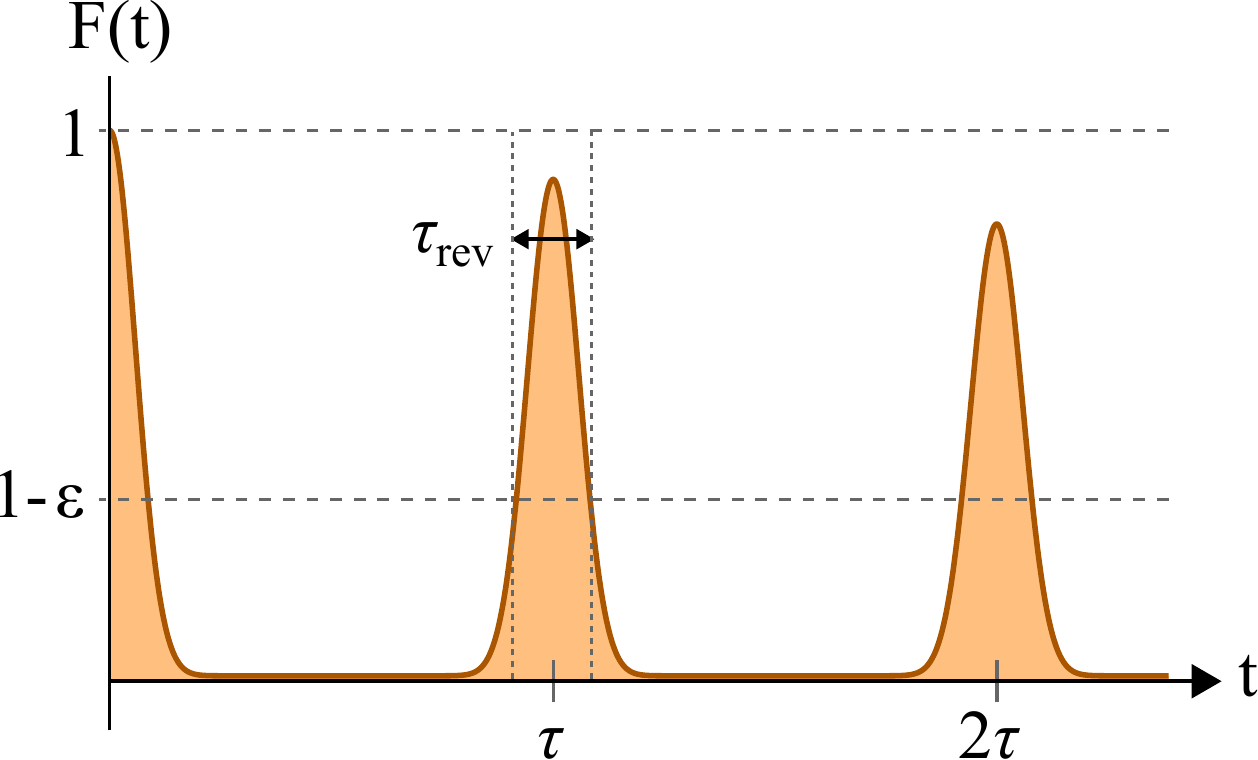}
    \caption{Illustration of the typical behaviour of the fidelity of a system with revivals. According to Theorems~\ref{th:fidbound} and \ref{th:lieb}, the width of the peaks $\tau_{\text{rev}}$ must decrease with system size.}
    \label{fig:revivals}
  \end{figure}

The second result utilizes the Lieb-Robinson bound \cite{lieb1972finite} to show that a short time after an initial product state is prepared (or equally, after a perfect revival), its overlap with the initial state has to be sub-exponentially small in the system size.
For simplicity, we formulate and prove this result only in the case of a $D$-dimensional cubic lattice of side-length $L$ and with a translationally invariant initial state and Hamiltonian. We emphasize, however, that a similar argument applies to any regular $D$-dimensional lattice and also for initial states that are only translationally invariant with a higher period than the lattice spacing.
\begin{theorem} \label{th:lieb}
Consider the translationally invariant initial state $\ket{\Psi}=\ket{\psi}^{\otimes N}$ on the cubic lattice $\Lambda=\mathbb Z_L^D$ evolving under a strictly local, uniformly bounded, and translationally invariant Hamiltonian $H$. 
Define
\begin{align}
    k(t) = -\log\left(\bra{\Psi} \rho_x(t)\otimes\mathbf 1 \ket{\Psi} \right),
\end{align}
where $\rho_x(t) = \Tr_{\{x\}^c}[U_t \proj{\Psi}U_t^\dagger]$ is the reduced density matrix of an arbitrary site $x$ at time $t$. If $\ket{\Psi}$ is not an eigenstate, then for any $\delta>0$ there exists a time $0<\tau<\delta$ such that $k(\tau)>0$. For any such fixed time $\tau$ and for large enough $L$, we have
\begin{align}
    F(\tau)^2 &\leq \mc O\left(\exp\left(-\frac{1}{4}\left[L^D k(\tau)\right]^{1/{1+D}}\right)\right).
\end{align}
\end{theorem}

  \begin{figure}[t]
    \includegraphics[width=\linewidth]{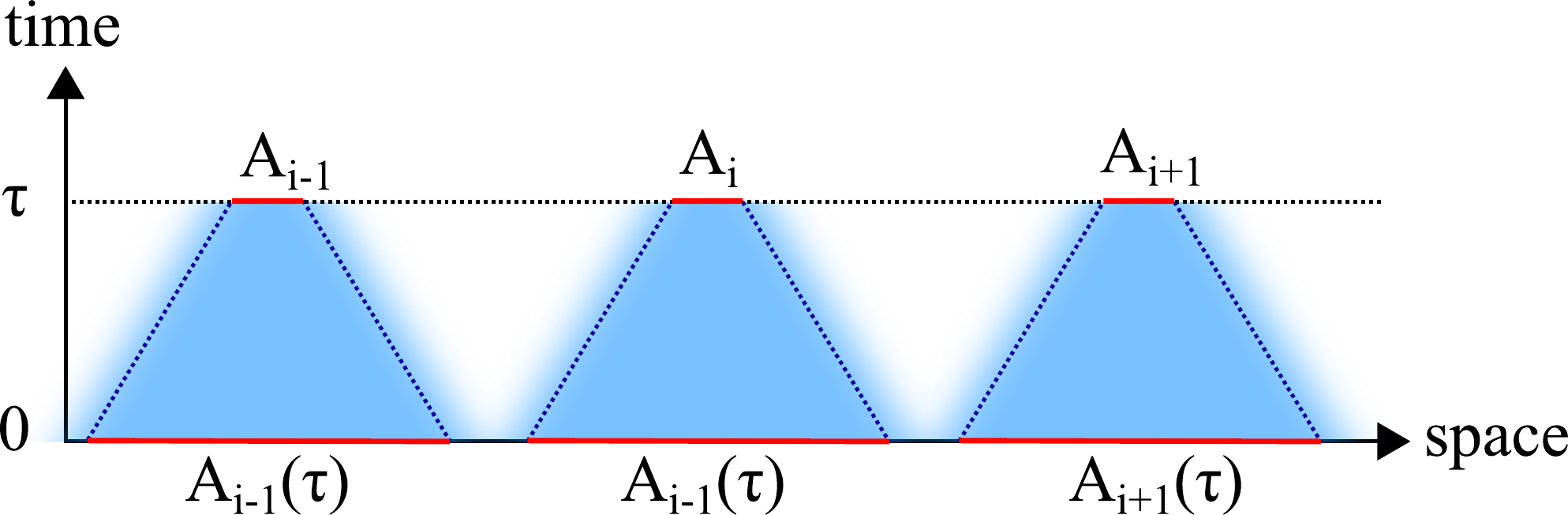}
    \caption{Schematic illustration for the idea behind Theorem~\ref{th:lieb}: An operator located in region $A_i$ at time $t=\tau$ can (approximately) only be influenced by the region $A_i(\tau)$ at time $t=0$ due to the past Lieb-Robinson ``light-cone". If the quantum state of the regions $A_i(\tau)$ factorizes at time $t=0$, it follows that correlation functions between operators supported in the regions $A_i$ (approximately) factorize at time $\tau$. Since translationally invariant product states have exponentially small overlap, one hence expects that the states at times $t=0$ and $t=\tau$ have exponentially small fidelity (as evidenced by concrete models, see Appendix~\ref{app:example}).}
    \label{fig:lieb}
\end{figure}
The theorem says that, whenever $k(\tau)>0$, the fidelity between $\ket{\Psi}$ and the time-evolved state $U_\tau \ket{\Psi}$ is sub-exponentially small in the linear size of the system $L$. Furthermore, from a perturbative expansion one quickly finds that for small $\tau$ we have $k(\tau)=\mc O(\tau^2)$. 
The proof of Theorem~\ref{th:lieb} is relatively involved and presented in Appendix~\ref{sec:app:LR}. 
However, the idea behind it is simple and sketched in Fig. \ref{fig:lieb}. 

\section{Conclusion}
We have derived general results on the energy spectrum and the entanglement of (approximate) energy eigenstates for systems that show revivals. Most importantly, our results show that the presence of ``quantum many-body scars" with small amounts of entanglement of order $\log(N)$ is a necessary consequence of the existence of revivals of a low entangled state with a revival time that is at most $\mc O(\poly(N))$.  This explains why this scaling behaviour has been found in the concrete models studied so far \cite{Turner2018a,vafek2017entanglement,moudgalya2018entanglement,choi2019emergent,schecter2019weak,Chattopadhyay2019,Iadecola2019kinetic}. 
 One drawback of our results is that they only show a $\mathcal{O}(\log (N))$ scaling for R\'enyi entanglement entropies of orders $\alpha >1$, while for smaller $\alpha$ we can only show the upper bound $\mathcal{O}(\sqrt{N | \partial A |})$. 
 
 While it is often found in practice that the von Neumann entropy and the higher order R\'enyi entanglement entropies show a similar scaling behaviour, this is not always the case. In particular our bounds on the R\'enyi entanglement entropies do not guarantee the existence of an efficient description in terms of matrix product states (MPS) \cite{schuch2008entropy} (although a $\mathcal{O}(\log (N))$ bound on the von~Neumann would not imply this either). Indeed, it is known \cite{wilming2018entanglement} that there exist states with both \emph{i)} an arbitrarily large overlap with a product state and \emph{ii)} a volume-law scaling of the von Neumann entropy, while all R\'enyi entropies of order $\alpha>1$ are bounded by a constant (dependend on $\alpha$). 
It is an interesting open problem to find arguments for bounding the R\'enyi entropy with $\alpha<1$ by $\mathcal{O}(\log (N))$, which would guarantee an efficient MPS description of the scar states from dynamical considerations alone. A further interesting open problem is to understand whether the emergent (approximate) SU(2) representations that are connected to quantum many-body scars in concrete models \cite{moudgalya2018entanglement,choi2019emergent,schecter2019weak,Chattopadhyay2019,Iadecola2019kinetic} can be derived from general arguments. 
Finally, it would be interesting to see whether results similar to those in the case of approximate revival of the initial state can also be derived for approximate revivals of (generic) expectation values of observables. This would also be interesting from the point of view of bounding equilibration time-scales in interacting quantum many-body system, a problem where relatively little rigorous progress has been made so far (see, for example, Refs.~\cite{gogolin2016equilibration,Wilming17,deOliveira2018equilibration,Wilming2018equilibration} and references therein for recent discussions of this problem). In particular, it is an interesting open problem whether $\epsilon$-approximate revivals of local observables and at early times are possible in \emph{entanglement-ergodic} systems \cite{wilming2018entanglement}, where all energy eigenstates at positive energy density fulfill weak volumes laws of entanglement.

\emph{Acknowledgements.} The authors acknowledge useful discussions with Cheng-Ju Lin.
H.~W. acknowledges support from the Swiss National Science Foundation through SNSF project
No. 200020\_165843 and through the National Centre of Competence in Research \emph{Quantum Science and Technology} (QSIT). AA is supported by the Canadian Institute for Advanced Research, through funding provided to the Institute for Quantum Computing by the Government of Canada and the Province of Ontario. This research was supported in part by Perimeter Institute for Theoretical Physics. Research at Perimeter Institute is supported in part by the Government of Canada through the Department of Innovation, Science and Economic Development Canada and by the Province of Ontario through the Ministry of Colleges and Universities. Below we report the estimated CO$_2$ emissions for transport  due to this project (template from  \href{https://scientific-conduct.github.io}{scientific-conduct.github.io}).
\begin{center}
\begin{tabular}[b]{l c}
\hline
\textbf{Transport} & \\
\hline
Total CO$_2$-Emission For Transport [$\mathrm{kg}$] & 2910\\
Were The Emissions Offset? & \textbf{Yes}\\
\hline
Total CO$_2$-Emission [$\mathrm{kg}$] & 2910\\
\hline
\hline
\end{tabular}
\end{center}
\bibliography{scars,references}

\widetext

\appendix

\section{Technical proofs of main results}
\subsection{Proof of Inequality~\eqref{eq:revivals}}
\label{sec:revivals}
Here we show inequality~\eqref{eq:revivals}, using the notation $\ket{\Psi(t)}=U_t\ket{\Psi}$. 
First, from the triangle inequality and the unitary invariance of the trace-norm, we find
\begin{align}
    \norm{\proj{\Psi(0)}-\proj{\Psi(m\tau)}}_1 \leq m \norm{\proj{\Psi(0)}-\proj{\Psi(\tau)}}_1. 
\end{align}
We now make use of the Fuchs-van de Graaf inequalities \cite{Fuchs1999}
\begin{align}
    1 - F(\rho,\sigma) \leq \frac{1}{2}\norm{\rho-\sigma}_1 \leq \sqrt{1-F(\rho,\sigma)^2},
\end{align}
where $F(\rho,\sigma) = \norm{\sqrt{\rho}\sqrt{\sigma}}_1$ is the fidelity between two quantum states. In our case we have
\begin{align}
    F(t) = |\braketn{\Psi(t)}{\Psi(0)}| = F(\proj{\Psi(t)},\proj{\Psi(0)}).
\end{align}
Using $F(\tau)^2 =(1-\epsilon)^2 \geq 1 - 2\epsilon$, we thus find
\begin{align}
    1-F(m\tau) \leq m\sqrt{1- F(\tau)^2} \leq m \sqrt{2\epsilon},
\end{align}
which proves the claim.

\subsection{Proof of Lemma \ref{le:upperbound}}
\label{app:upperbound}

We first need the Berry-Esseen theorem for local Hamiltonians from \cite{brandao2015equivalence},  which reads as follows.
\begin{theorem}
(Lemma 8 of \cite{brandao2015equivalence})\label{th:berry}
	Let $\ket{\Psi}$ be a state with a finite correlation length, energy variance $\sigma$, and a local Hamiltonian with uniformly bounded local terms, of a system of $N$ particles on a $D$-dimensional lattice. Given the cumulative function
	\begin{equation}
	J(x)=\sum_{E_i \le x} |c_i|^2
	\end{equation}	
	and the Gaussian cumulative function
	\begin{equation}
	G(x)=\int_{-\infty}^x \frac{\text{d}t}{\sqrt{2\pi \sigma^2}} e^{\frac{-(t-\langle H \rangle)^2}{2 \sigma^2}},
	\end{equation}
	then 
	\begin{equation}
	\sup_x |J(x)-G(x)| \le C \frac{\log^{2D}(N)}{s^3 \sqrt{N}},
	\end{equation}
where $C$ is a constant and $s=\frac{\sigma}{\sqrt{N}}$.
\end{theorem}
The proof of Lemma~\ref{le:upperbound} follows straightforwardly from this result.
Let us first recall the definition of $p(\Delta_\delta(l))$ from the main text,
\begin{align}
	p(\Delta_\delta(l)) = \sum_{i:E_i \in \Delta_\delta(l)} |c_i|^2, 
\end{align}
where $\Delta_\delta(l) = \frac{2\pi l+ \alpha(\tau)}{\tau} + \left[-\frac{\delta}{\tau}, \frac{\delta}{\tau}\right]$. Using the notation of Theorem \ref{th:berry}, we can write
\begin{align}
	p(\Delta_\delta(l)) = J\left(\frac{2\pi l+ \alpha(\tau)}{\tau} + \frac{\delta}{\tau}\right)-J\left(\frac{2\pi l+ \alpha(\tau)}{\tau} - \frac{\delta}{\tau}\right).
\end{align}
Now, using the triangle inequality and the Berry-Esseen bound, 
\begin{align}
\vert J(x+y)-J(x)\vert &\le \vert G(x+y)-G(x) \vert + \vert J(x+y)-G(x+y) +G(x)-J(x) \vert \\&\le \vert G(x+y)-G(x) \vert + 2 \sup_x |J(x)-G(x)|
\\&\le \vert G(x+y)-G(x) \vert + K \frac{\log^{2D}(N)}{\sqrt{N}},
\end{align}
where $K \ge 2 C/s^3$.
By definition we have that, for all $x\ge 0$, 
\begin{equation}
    \vert G(x+y)-G(x) \vert = \int_{x}^{x+y} \frac{\text{d}t}{\sqrt{2\pi \sigma^2}} e^{\frac{-(t-\langle H \rangle)^2}{2 \sigma^2}}  \le \int_{-y/2+\langle H \rangle}^{y/2+\langle H \rangle} \frac{\text{d}t}{\sqrt{2\pi \sigma^2}} e^{\frac{-(t-\langle H \rangle)^2}{2 \sigma^2}}= \int_{-y/2}^{y/2} \frac{\text{d}t}{\sqrt{2\pi \sigma^2}} e^{\frac{-t^2}{2 \sigma^2}} \le \frac{ y}{2\sigma},
\end{equation}
where in the last step we have used that 
$\text{Erf}[x]\equiv \frac{1}{\sqrt{\pi}}\int_{-x}^{x}\text{d}t e^{-t^2}\le \sqrt{2} x$.
Setting $y=\frac{2\delta}{\tau}$ completes the proof.

\subsection{Proof of Theorem \ref{th:S0scar}} \label{app:upperbound2}

We start with the assumption of perfect revivals $F(\tau)=F(0)$, which implies that 
\begin{equation}
    \ket{\Psi}=\sum_i c_i \ket{E_i} \,,\, \,\, E_i=\frac{2\pi l+\alpha(t)}{\tau},
\end{equation}
where $l \in \mathbb{Z}$. Since the Hamiltonian is bounded, we know that the number of such levels is at most $\frac{hN \tau}{2\pi}$, so that $\ket{\Psi}$ has support on at most that many levels. Let us now define the following operator $K_i$ for the eigenstate $\ket{E_i}$:
\begin{equation}
\label{eq:exASSP2}
K_i(H) = \prod_{j\neq i}\left (\br{\id - \frac{\br{H-E_i}^2}{\br{E_j-E_i}^2}} \right ).
\end{equation}
This is a polynomial of the Hamiltonian of degree at most $\frac{hN \tau}{\pi}$, and is such that $K_i(H)\ket{\Psi}= c_i \ket{E_i}$.

Now fix a bipartition $A\cup \bar{A}$ of the lattice. The boundary of a set $A$, denoted $\cB(A)$, is the set of local terms $h_x$ which are supported on both $A$ and $\bar{A}$. Assume the following regularity conditions for it, parametrized by an integer $m$. 
\begin{definition}
\label{def:regular}
Let $m$ be an integer. The bipartition $A\cup \bar{A}$ is said to be $m$-regular if the following holds.
\beit
\item There exist $2m$ non-empty concentric sets $A_{-m}, \ldots A_{-2}, A_{-1}, A_1, A_2, \ldots A_m$ satisfying $A \subset A_1 \subset A_2 \subset \ldots A_m$ and $A \supset A_{-1} \supset A_{-2}\supset\ldots A_{-m}$. Let $A_0=A$.
\item The number of spins in $A_m\setminus A$ and $A \setminus A_{-m}$ is at most $10 m |\partial A|$.
\item Every local term $h_x$ belongs to at most one $\cB(A_q)$.
\enit
\end{definition}
\noindent Several natural partitions, such as rectangular, vertical and circular, satisfy these conditions, whenever $|\partial A|= \Omega(m)$ (i.e., for sufficiently large regions). 

Given this definition, we now proceed to bound the Schmidt rank of polynomials of the Hamiltonian, using a result from \cite{arad2013area}.

\begin{lemma}
\label{lem:SRbound}
Fix an integer $m>0$. For any $m$-regular bipartition of the lattice, it holds that $$\sr{H^{\ell}}\leq \br{2\ell N d^b}^{\frac{2\ell}{m} + 10m|\partial A|},$$
\end{lemma} 
where $\sr{X}$ denotes the Schmidt rank of an operator $X$ with respect to the bipartition.
\begin{proof}
Following the regularity condition (Definition \ref{def:regular}) on $A$, we write $H=\sum_{q=-m}^{m+1} H_q$, where
$$H_q = \sum_{x: \supp(h_x)\subset A_q\setminus A_{q-1}} h_x + \sum_{x\in \cB(A_{q-1})}h_x$$
for $q\in \{-m+1, \ldots m\}$, 
$$H_{-m} = \sum_{x: \supp(h_k) \subset A_{-m}}h_x$$
and
$$H_{m+1} = \sum_{x: \supp(h_x)\in \cL\setminus A_{-m}} h_x.$$
Following \cite{anshu2019entanglement}, we can view $H$ as an effective one dimensional Hamiltonian. It was argued in \cite{arad2013area} that $H^{\ell}$ can be expanded as a linear combination of at most ${\ell+m \choose m}\cdot {2\ell \choose \frac{2\ell}{m}}$ multinomials of $\{H_q\}_{q=-m}^{m+1}$, such that each multinomial has a Schmidt rank at most $\br{Nd^b}^{\frac{\ell}{m}}$ across some bipartition $A_q\cup \bar{A_q}$ (with $-m<q < m$). Since the number of spins between $A_q$ and $A$ is at most $10m|\partial A|$ (recall Definition \ref{def:regular}), each such multinomial has a Schmidt rank of at most $\br{Nd^b}^{\frac{\ell}{m}}\cdot d^{10m|\partial A|}$ across the bipartition $A\cup \bar{A}$. Thus,
$$\sr{H^{\ell}}\leq {\ell+m \choose m}\cdot {2\ell \choose \frac{2\ell}{m}}\cdot \br{Nd^b}^{\frac{\ell}{m}}\cdot d^{10m|\partial A|}\leq \br{2\ell N d^b}^{\frac{2\ell}{m} + 10m|\partial A|}.$$
This completes the proof.
\end{proof}

On the other hand, assumption \ref{ass:entanglement} from the main text implies that there exists a product state $\ket{\psi_A}\otimes \ket{\psi_{A'}}$ and an operator $K_{\Psi}$ of Schmidt rank $\sr{K_{\Psi}}\le \chi^{\vert \partial A \vert }$ such that $K_{\Psi} \ket{\psi_A}\otimes \ket{\psi_{A'}}= \ket{\Psi}$. 

We now proceed to prove Theorem \ref{th:S0scar}. Notice that Eq. \ref{eq:exproj} states that $\ket{E_i}$ can be obtained from $\ket{\Psi}$ by applying a polynomial in $H$ of degree at most $\frac{hN \tau}{\pi}$. Applying Lemma \ref{lem:SRbound} and choosing $\ell= \frac{hN \tau}{\pi}, m=\sqrt{\frac{\ell}{|\partial A|}}$ yields
\begin{equation}
    \log{\sr{K_i}} \le 12 \sqrt{\frac{hN \tau |\partial A|}{\pi}} \log{\left(\frac{2hN^2 \tau d^b}{\pi}\right)}\le 7\sqrt{hN \tau |\partial A|} \log{\left(hN^2 \tau d^b\right)}
\end{equation}
for every $m$-regular partition. 
With this, a bound on the R\'enyi-0 entropy follows
$$\entz{\Tr_{\bar{A}}\br{\ket{E_i}\bra{E_i}}} \leq  \log{\sr{K_i}}  +  \log{\sr{K_{\Psi}}} \le 7\sqrt{hN \tau |\partial A|} \log{\left(hN^2 \tau d^b\right)} + \vert \partial A \vert \log{\chi} $$
This completes the proof.

\subsection{Proof of Theorem \ref{th:fidbound}}

The proof follows an argument first made in \cite{Malabarba14,Garcia-PintosPRX2017}, combined with the Berry-Esseen theorem (see section~\ref{app:upperbound}). 
Since the integrand of the LHS is positive, we have that 
\begin{align}
 \int_0^T \frac{\text{d}t}{T} \left\vert F(t) \right \vert^2 &\le \frac{5}{4} \int_0^T \frac{T\text{d}t}{T^2+(t-\frac{T}{2})^2}  \left\vert F(t) \right \vert^2 \\ \nonumber &=  \sum_{l,m}  |c_l|^2 \vert c_m \vert^2   \frac{5}{4}  \int_0^T \frac{T\text{d}t}{T^2+(t-\frac{T}{2})^2}   e^{-i t ( E_l-E_m ) } \\ \nonumber &\le  \sum_{l,m}  |c_l|^2 \vert c_m \vert^2   \frac{5}{4} \left \vert \int_0^T \frac{T\text{d}t}{T^2+(t-\frac{T}{2})^2}   e^{-i t ( E_l-E_m ) } \right \vert \\
 &=  \frac{5 \pi}{4} \sum_{l , m}  |c_l|^2 \vert c_m \vert^2   e^{- T \vert  E_l-E_m \vert } \nonumber
\end{align}
Appendix B of \cite{Malabarba14} (and also Appendix B of \cite{Garcia-PintosPRX2017}) then shows that
\begin{equation}
 \sum_{l , m}  |c_l|^2 \vert c_m \vert^2   e^{- T \vert  E_l-E_m \vert }   \le 
	4 \max_E \sum_{  E_l \in \{E,E+1/T\}} |c_l|^2  \equiv 4 \xi(1/T).
\end{equation}
Next, we again use the Berry-Esseen theorem (Theorem \ref{th:berry}). 
Note that we can write 
\begin{equation}\label{eq:xi}
  \xi(1/T)=\max_E J(E+1/T)-J(E).  
\end{equation}
 Let us define $E^*$ as the solution of the optimization in Eq. \eqref{eq:xi}. This yields
	\begin{align}
	\big \vert \xi(1/T) - G\left(E^*+\frac{1}{T}\right) +G(E^*) \big \vert &\le 2 \sup_x |J(x)-G(x)|\\& \nonumber \le 2 C \frac{\log^{2D}(N)}{s^3 \sqrt{N}},
	\end{align}
where in the first line we use the triangle inequality, and in the second the Berry-Esseen theorem. Thus
	\begin{align}
\int_0^T \frac{\text{d}t}{T} \left\vert F(t) \right \vert^2 &\le 5\pi \left(G\left(E^*+\frac{1}{T}\right) -G(E^*)\right) + 10\pi C  \frac{\log^{2D}(N)}{s^3\sqrt{N}} \\
&\le 5\pi \text{Erf}[1/(\sqrt{2 \sigma}T) ] + 10\pi C  \frac{\log^{2D}(N)}{s^3\sqrt{N}},
	\end{align}
	where $\text{Erf}[x]\equiv \frac{1}{\sqrt{\pi}}\int_{-x}^{x}\text{d}t e^{-t^2}$.
The result now follows from the fact that $\text{Erf}[x] \le \sqrt{2} x $ and setting $K'\ge 10\pi C/s^3$.

	\subsection{Proof of Theorem \ref{th:lieb}}
\label{sec:app:LR}
The proof of Theorem~\ref{th:lieb} relies on Lieb-Robinson (LR) bounds \cite{lieb1972finite}, which apply under our assumptions on the Hamiltonian stated in the main text. LR bounds can be stated in different ways, see Ref.~\cite{Kliesch2014a} for a review. For our purposes, the following formulation will be most suitable.
To state it, we have to set up some notation. For any region $X$ of the lattice $\Lambda=\mathbb Z_L^D$, we define 
\begin{align}
    H_X &= \sum_{x: \mathrm{supp}(h_x)\subseteq X} h_x
\end{align}
as the sum of Hamiltonian terms supported in the region $X$.
For any region $X$ we define the complement $X^c=\Lambda\setminus X$ and $|X|$ to be the number of lattice sites contained in $X$. For any two regions of the lattice $X,Y$  we denote by $d(X,Y)$ the lattice distance between the regions. 
For any region $X$ we further define $U^X_t$ as the unitary propagator for time $t$ under the Hamiltonian $H_X$. Given our conventions, we thus have $U_t = U^\Lambda_t$. 
Now let $A$ be a local observable. By abuse of notation, we also denote by $A$ its supporting region on the lattice and hence by $|A|$ the corresponding number of sites of the lattice on which it acts. We further define
\begin{align}
    A^X(t) = {U^X_t}^\dagger \hat A U^X_t,\quad A_\Lambda(t) = A(t). 
\end{align}
We are now in position to state the LR bounds that we will use. 
\begin{lemma}[Lieb-Robinson bounds]
Let $A$ be a local observable and $H$  a strictly local, and uniformly bounded Hamiltonian. Then there exist constants $K_{\LR},v_{\LR}\geq 0$ such that for all $X$ with $l:=d(A,X^c)\geq 2D-1$ we have
\begin{align}
    \norm{A^X(t) - A(t)}\leq \norm{A}K_{\LR}  l^{D-1}\e^{v_{\LR} t - l}.
\end{align}
\end{lemma}
The Lieb-Robinson bounds tell us that we can approximate the time evolution of a local observable $A$ by time evolution constrained to a neighbourhood $X$ around it as long as the distance $l$ from $A$ to the complement of $X$ is much larger than $v_{\LR}t$. 
In turn this implies that regions on a lattice that are a distance $l$ apart cannot build up significant correlations within a time much smaller than $l/v_{\LR}$.

For now, we will keep the proof slightly more general than the statement in the theorem and consider an initial state $\ket{\Psi} = \otimes_x \ket{\psi_x}$ that need not be translationally invariant. We then specialize to the latter case towards the end of the proof. In the following we write $\ket{\Psi(\tau)}=U_\tau \ket{\Psi}$. 
First, using that $\ket{\Psi}=\ket{\Psi(0)}$ is a product state, we find for any region $A$:
\begin{align}
\vert \braketn{\Psi(\tau)}{\Psi(0)} \vert^2
= \bra{\Psi(\tau)} \otimes_{x\in\Lambda}\proj{\psi_x}\ket{\Psi(\tau)} 
\leq \bra{\Psi(\tau)} \otimes_{x\in A} \proj{\psi_x} \otimes \mathbf 1_{A^c} \ket{\Psi(\tau)}.
\end{align}
Viewing $A_x=\proj{\psi_x}\otimes \mathbf 1_{\Lambda\setminus \{x\}}$ as a local observable supported on site $x\in A$ and $A=\otimes_{x\in A}\proj{\psi_x}\otimes\mathbf 1_{A^c}$ as one supported on region $A$, we can then make use of the Heisenberg picture to get
\begin{align}
    \vert \braketn{\Psi(\tau)}{\Psi(0)} \vert^2 \leq \bra{\Psi(0)} A(\tau) \ket{\Psi(0)}.
\end{align}
We now fix the region $A$ to consist of a sub-lattice of sites, all of which are a distance $2(l+1)+r$ apart from each other, where $r$ is the maximum diameter of the support size of the Hamiltonian terms:
\begin{align}
        r = \max_{x\in\Lambda}|\mathrm{diam}(\mathrm{supp}(h_x))|.
\end{align}
The distance $l$ will be fixed later. We define $B_x(l)$ to be an $l$-neighbourhood of $x$,
\begin{align}
    B_x(l) = \{ y\in \Lambda\ \vert\ d(x,y)\leq l\},
\end{align}
and set $X = \cup_{x\in A} B_x(l-1)$. With this choice we have $d(A,X^c)=l$ and
\begin{align}
    U^X_t = \prod_{x\in A} V^x(t),
\end{align}
where $V^x(t)$ is only supported within $B_x(l-1)$, which implies
\begin{align}
    A^X(t) = \prod_{x\in A} A^x(t).
\end{align}
Using the LR-bounds and that $\ket{\Psi(0)}$  is a product state, we then find
\begin{align}
    \bra{\Psi(0)} A(\tau) \ket{\Psi(0)} \leq \bra{\Psi(0)} A^X(\tau) \ket{\Psi(0)} + K_{\LR}  l^{D-1}\e^{v_{\LR} \tau - l} = \prod_{x\in A} \bra{\Psi(0)}A_x^X(\tau)\ket{\Psi(0)}+ K_{\LR}  l^{D-1}\e^{v_{\LR} \tau - l}.
\end{align}
We can now make use of the LR-bounds again to approximate each factor:
\begin{align}
    \bra{\Psi(0)}A^X_x(\tau)\ket{\Psi(0)} \leq \bra{\Psi(0)}A_x(\tau)\ket{\Psi(0)} + K_{\LR}  l^{D-1}\e^{v_{\LR} \tau - l} = \bra{\Psi(\tau)}A_x\ket{\Psi(\tau)}+ K_{\LR}  l^{D-1}\e^{v_{\LR} \tau - l}. 
\end{align}
Using that $\norm{A(t)}=\norm{A^X(t)}=\norm{A_x^X(t)} = \norm{A_x(t)}=1$, we can then bound
\begin{align}
\bra{\Psi(0)} A(\tau) \ket{\Psi(0)} \leq \prod_{x\in A} \bra{\Psi(0)}A_x(\tau)\ket{\Psi(0)}+ 2|A|K_{\LR}  l^{D-1}\e^{v_{\LR} \tau - l}.
\end{align}
However, we also have
\begin{align}
    \bra{\Psi(0)}A_x(\tau)\ket{\Psi(0)} = \bra{\Psi(\tau)} \proj{\psi_x}\otimes \mathbf 1_{\{x\}^c}\ket{\Psi(\tau)} = \Tr\left[\rho_x(\tau) \proj{\psi_x}\right] =: \exp(-k_x(\tau)). 
\end{align}
Putting the bounds together and using the assumption
\begin{align}
    k(\tau) = \min_{x\in \Lambda} k_x(\tau) > 0,
\end{align}
we thus find
\begin{align}
    \vert \braketn{\Psi(\tau)}{\Psi(\tau)} \vert^2 \leq \exp(-k(\tau)|A|) + 2|A|K_{\LR}  l^{D-1}\exp(v_{\LR} \tau - l).
\end{align}

We now choose $l=\left(L^D k(\tau)\right)^{1/(1+D)}$. For large enough $L$, we then find
\begin{align}
    |A| &\leq \frac{1}{2}\left(\frac{L}{l}\right)^D = \frac{1}{2}\left(\frac{L^D}{k(\tau)}\right)^{\frac{D}{1+D}},\\
    |A| &\geq \frac{1}{4}\left(\frac{L^D}{k(\tau)}\right)^{\frac{D}{1+D}}.
\end{align}
This leads to
\begin{align}
    \vert \braketn{\Psi(\tau)}{\Psi(0)}\vert^2 &\leq \left[1 +  K_{\LR}\left(\frac{L^D}{k(\tau)}\right)^{\frac{D}{1+D}}\left(L^D k(\tau)\right)^{D-1/(D+1)}\e^{v_{\LR}\tau}\right] \exp\left(-\frac{1}{4}\left(L^D k(\tau)\right)^{1/(1+D)}\right)\\
    &= \left[1 +  K_{\LR}\left(\frac{N^{2D-1}}{k(\tau)}\right)^{\frac{1}{D+1}} \right] \exp\left(-\frac{1}{4}\left(N k(\tau)\right)^{1/(1+D)}\right)\\
    &= \mc O\left(\exp\left(-\frac{1}{4}\left(N k(\tau)\right)^{1/(1+D)}\right)\right).
\end{align}
What is left is to show that for any $\delta>0$ there exists a $\tau<\delta$ such that $k(\tau)>0$. To do this, we now make use of translational invariance, so that $k(\tau)=k_x(\tau)$ for any $x\in\Lambda$. Suppose now contrarily that there exists a $\delta>0$ and $k(\tau)=0$ for all $\tau<\delta$. This means that $k(\tau)$ is constant over an open interval. But since on any finite system $k(\tau)$ is an analytic function, it then has to be constant. This in turn implies that
\begin{align}
    \rho_x(\tau) = \proj{\psi}_x
\end{align}
for all $\tau$, which implies that the initial state is an eigenstate. This finishes the proof of Theorem~\ref{th:lieb}.
We emphasize that we only used translational invariance of the initial state to argue that $k(\tau)>0$. It should be clear from the argument given above that it can be generalized to situations where, for example, the initial state is translationally invariant with a higher period, or is only a product state after neighboring spins are blocked together.

\section{Comparing the bounds with previous results}\label{app:example}

To illustrate the tightness of our bounds, we compare our results with those of a recently found model with quantum scars and perfect revivals in Ref.~\cite{schecter2019weak}. The model is the spin-1 XY model in a hypercubic lattice with $N=L^D$ particles and Hamiltonian
\begin{equation}
    H=\sum_{\langle ij \rangle} (S^{x}_iS^x_j+S^{y}_i S^y_j) + h \sum_i S^z_i + D \sum_i (S^z_i)^2,
\end{equation}
where $S^{\alpha}_i$ are the spin-1 operators at site $i$. In Ref.~\cite{schecter2019weak} it was found that this Hamiltonian has $N+1$ eigenstates $\ket{\mathcal{S}_n}$ with $n \in \{0,...,N\}$ which form a representation of $\text{SU}(2)$ and have equally spaced energies $E_n=h(2n-N)+ND$. Moreover, there exists a particular product state $\ket{\Psi_0}=\bigotimes_i \ket{\psi_i}$, the so-called ``nematic N\'eel" state, which is such that
\begin{equation}
    \vert F(t) \vert^2=\vert \bra{\Psi_0} e^{-iHt} \ket{\Psi_0}\vert^2 =\cos^{2N}(ht),
\end{equation}
that is, it exhibits perfect revivals at periods of $\pi / h$, with a weight suppressed exponentially with the system size. This initial product state can be written as
\begin{equation}
    \ket{\Psi_0}=\sum_{n=0}^N c_n \ket{\mathcal{S}_n}, \,\, c_n^2=\frac{1}{2^N}\binom{N}{n},
\end{equation}
so that $\langle H \rangle = ND$ and $\sigma=h \sqrt{N}$, which thus fulfills assumptions \ref{ass:entanglement} and \ref{ass:clustering} from the main text.
One can easily calculate that for any $T=\pi l /h$ we have
\begin{equation}
	\int_0^T \frac{\text{d}t}{T}\vert F(t) \vert^2= \frac{\left (N -\frac{1}{2}\right)!}{\sqrt{\pi} N!}=(\pi N)^{-1/2}+\mathcal{O}(N^{-3/2}),
\end{equation}
which shows that the scaling of Theorem \ref{th:fidbound} is close to optimal.

For a bi-partition of the lattice $N=N_A+N_B$ with $N_A\le N_B$, the scar eigenstates have a Schmidt decomposition $\ket{\mathcal{S}_n}=\sum_{k=0}^K \sqrt{\lambda_k^{(n)}}\ket{i_{k,A}^{(n)}}\otimes \ket{i_{k,B}^{(n)}}$ where $K=\max \{n, N_A \}$. The coefficients are calculated to be
\begin{equation}
    \lambda^{(n)}_k=\frac{\binom{N_A}{k} \binom{N_B}{n-k}}{\binom{N}{n}}.
\end{equation}
Let us choose $N_A=N_B=N/2$. The R\'enyi-$\infty$ entropy can now be easily obtained by noting that for all $n$, the largest Schmidt coefficient is given by
\begin{equation}
    \lambda^{(n)}_{\text{max}}=\frac{\binom{N/2}{n/2}^2 }{\binom{N}{n}}.
\end{equation}
Let us now do the change of variables $n=b N$, so that $b$ is a $\mathcal{O}(1)$ number for the $\mathcal{O}(N)$ eigenstates in the bulk of the spectrum. For large $N$, using Stirling's approximation, we find that
\begin{align}
\binom{N/2}{bN/2}^2&\simeq \frac{1}{\pi N b(1-b)}\left(\frac{N}{2} \frac{1}{\left(\frac{bN}{2}\right)^{b}\left(\frac{(1-b)N}{2}\right)^{(1-b)} } \right)^N \\
\binom{N}{bN}&\simeq \frac{1}{\sqrt{2 \pi b (1-b) N}} \left(N \frac{1}{\left(bN\right)^{b}\left((1-b)N\right)^{(1-b)} } \right)^N,
\end{align}
which leads to
\begin{equation}
\lambda_{\max}^{(n)}\simeq \sqrt{\frac{2}{\pi b(1-b) N}},
\end{equation}
and therefore
\begin{equation}
S_\infty \equiv -\log \lambda_{\max}^{(n)} \simeq \frac{1}{2}\log N + \frac{1}{2}\log (\pi b(1-b)/2).
\end{equation}
This, together with the inequalities for R\'enyi entropies $S_\alpha \geq S_\infty$ for all $\alpha\geq 0$ and
\begin{equation}
\frac{\alpha-1}{\alpha} S_\alpha    \le S_{\infty} \,\,\,\, \forall \alpha \ge 1,
\end{equation}
implies that all the R\'enyi entropies with $\alpha>1$ of the $\mathcal{O}(N)$ eigenstates also scale logarithmically, and that the R\'enyi entropies for $0 \le \alpha \le 1$ scale at least logarithmically. Similar conclusions can likely be reached with further models in the literature such as \cite{choi2019emergent,Chattopadhyay2019}. In contrast, Corollary \ref{co:entanglement} only guarantees the existence of $\mathcal{O}(\sqrt{N}/\log^{2D}{(N)})$ eigenstates in which the R\'enyi entropies with $\alpha>1$ scale at most logarithmically.

In addition, in \cite{schecter2019weak} it is shown that at least one eigenstate in the middle of the spectrum has von Neumann entanglement entropy scaling as $S_1 \simeq \mathcal{O}\left(\log N \right)$, which suggests that Theorem \ref{th:S0scar} is not tight.

\end{document}